\documentclass[journal]{IEEEtran}
\usepackage[utf8]{inputenc}
\usepackage{mathtools}
\usepackage{graphicx}
\usepackage[table]{xcolor} 
\usepackage{xcolor}
\usepackage{amsthm}
\usepackage{algorithm}
\usepackage{algpseudocode}
\usepackage{amssymb}
\usepackage{cite}
\usepackage{multirow}
\usepackage{bm}
\usepackage{tablefootnote}
\usepackage{threeparttable}
\usepackage{braket}
\usepackage{booktabs}
\usepackage{hyperref}
\usepackage{caption}
\usepackage{hhline}

\usepackage{orcidlink}

\usepackage{tikz}
\usepackage{tikzpeople}
\usetikzlibrary{tikzmark, shapes, fit,backgrounds}
\usetikzlibrary{matrix, decorations.pathreplacing, angles, quotes}
\usetikzlibrary{patterns, calc, positioning}
\usetikzlibrary{arrows, arrows.meta}
\usetikzlibrary{shapes.multipart}
\usetikzlibrary{fit}
\usetikzlibrary{shapes.geometric, positioning}
\usetikzlibrary{decorations.pathmorphing}

\tikzset{meter/.append style={draw, inner sep=10, rectangle, font=\vphantom{A}, minimum width=30, scale=.7, path picture={\draw[black] ([shift={(.1,.3)}]path picture bounding box.south west) to[bend left=50] ([shift={(-.1,.3)}]path picture bounding box.south east);\draw[black,-{Latex[scale=.5]}] ([shift={(0,.1)}]path picture bounding box.south) -- ([shift={(.3,-.1)}]path picture bounding box.north);}}}
\tikzset{snake it/.style={decorate, decoration=snake}}

\newcommand{\Tr}{\normalfont{\mbox{Tr}}}

\newtheorem{theorem}{Theorem}

\newtheorem{remark}{Remark}

\theoremstyle{definition}

\title{Space-sharing and Singleton Bounds for Entanglement-assisted Classical Coding}
\author{
  Yuhang Yao${}^{\orcidlink{0000-0003-2347-0091}}$, \IEEEmembership{Graduate Student Member, IEEE},
  Tushita Prasad${}^{\orcidlink{0000-0002-9371-8534}}$,
  Markus Grassl${}^{\orcidlink{0000-0002-3720-5195}}$, \IEEEmembership{Senior Member, IEEE},\\
  Syed Jafar${}^{\orcidlink{0000-0003-2038-2977}}$, \IEEEmembership{Fellow, IEEE}, and 
  Hua Sun${}^{\orcidlink{0000-0001-8777-7987}}$,  \IEEEmembership{Senior Member, IEEE}%
\thanks{
Yuhang Yao and Syed Jafar are with University of California Irvine, USA (e-mail: \{yuhangy5,syed\}@uci.edu). Tushita Prasad and Markus Grassl are with University of Gdansk, Poland (e-mail: tushita.prasad@phdstud.ug.edu.pl, markus.grassl@ug.edu.pl). Hua Sun is with University of North Texas, USA (e-mail: Hua.Sun@unt.edu).
The work of Markus Grassl and Tushita Prasad is carried out under the
`International Centre for Theory of Quantum Technologies 2.0: R\&D
Industrial and Experimental Phase’ project (contract
no.~FENG.02.01-IP.05-0006/23). The project is implemented as part of
the International Research Agendas Programme of the Foundation for
Polish Science, co-financed by the European Funds for a Smart Economy
2021-2027 (FENG), Priority FENG.02 Innovation-friendly environment,
Measure FENG.02.01.
}%
}

\allowdisplaybreaks

\begin{document}
\maketitle
\begin{abstract}
Recent work has noted that a space-sharing argument proves the tightness of the entropic quantum Singleton bounds, which was left open in the literature
for various settings involving only-quantum messages, only-classical messages, or both classical and quantum messages. Focusing on the setting of entanglement-assisted classical coding (EACC), in this letter we first elaborate upon the space-sharing argument and the tight Singleton bound for this setting,  and then establish a new tight entropic Singleton bound for EACC codes with  entanglement assistance distributed across  a subset of  encoders when only local quantum operations are allowed at each encoder.
\end{abstract}

\section{Introduction}
An  $[n,k,d;c]_q$ entanglement-assisted classical communication (EACC) code  \cite{Grassl2025codes} allows sending a $k$-dit\footnote{A dit is a $q$-ary symbol. A qudit is a $q$-dimensional quantum system.} classical message over $n$ uses of a $q$-dimensional quantum channel, with assistance from $c$ maximally entangled qudit pairs shared in advance between the encoder and decoder, while tolerating $d-1$ erasures in the transmission. Similar to the Singleton bound  in classical coding theory, it follows from \cite{Grassl2025codes} and \cite{mamindlapally2023singleton} that  any $[n,k,d;c]_q$ EACC code must satisfy the bound,
\begin{align} \label{eq:singleton}
	k \leq (1+c/n)(n-d+1).
\end{align}
We refer to \eqref{eq:singleton} as the EACC Singleton bound.
Reference \cite{Grassl2025codes} proposes a family of EACC codes with good distance $d$ and dimension $q$, for the parameter regime
\begin{align}\label{eq:singleton2}
k\leq \max(n-2d+2+c,n-d+1),
\end{align}
and leaves open the question of tightness of
\eqref{eq:singleton}. Recently, it has been noted in
\cite{sun2025capacityerasurepronequantumstorage} that a
`space-sharing' argument suffices to show the tightness of
\eqref{eq:singleton} for EACC codes, along with the tightness of the
entropic Singleton bounds in
\cite{mamindlapally2023singleton,grassl2022entropic} that involve
generalizations to quantum communication. In this letter we first
elaborate upon the space-sharing argument and the tightness of
\eqref{eq:singleton} for EACC codes, and then show that
\eqref{eq:singleton2} is a tight entropic Singleton bound if the
entanglement is distributed across separate encoders.

\section{Problem Formulation}
\subsection{Notations}
Let $\mathbb{Z},$ $\mathbb{R}$, and $\mathbb{C}$ denote the sets of
integers, real numbers, and complex numbers, respectively. The
notation $\mathbb{Z}_{\geq a}$ ($\mathbb{R}_{\geq a}$) denotes the
integers (reals) greater than or equal to $a$. By $[n]$ we denote the
set $\{1,2,\cdots, n\}$ for $n\in \mathbb{Z}_{\geq 1}$.  For integers
$a$, and $b$, $\gcd(a,b)$ denotes the greatest common divisor of $a$
and $b$.  The cardinality of a finite set $\mathcal{S}$ is denoted by
$|\mathcal{S}|$.

For a quantum system $A$, we use $\mathcal{H}_A$ to denote its
associated Hilbert space, and $\mathcal{D}(\mathcal{H}_A)$ to denote
the set of density operators on $\mathcal{H}_A$.  The dimension of a
quantum system $A$ is the dimension of its associated Hilbert space
and is denoted as $|A|$. By default, in this letter a $q$-dimensional
quantum system is referred to as a qudit. The size of a quantum system
$A$ is measured as $\log_q |A|$ (qudits). A qudit with $q=2$ is called
a qubit.

We use $\mathbb{I}(x)$ to denote the indicator function, returning $1$
when the predicate $x$ is true and $0$ otherwise.  The family of all
$k$-element subsets of $\mathcal{S}$ is denoted by
$\binom{\mathcal{S}}{k}$ . Finally, $H(\cdot)$ and $I(\cdot;\cdot)$
denote Shannon entropy and mutual information for classical random
variables, and von Neumann entropy and mutual information for quantum
systems; conditional versions are defined analogously.

\subsection{Entanglement-assisted Classical Communication (EACC)}
In the problem of entanglement assisted classical communication
(EACC), a transmitter (Alice) wants to communicate a message $M$ to
the receiver (Bob), by utilizing a $q$-dimensional quantum erasure
channel $n$ times, along with $c$ $q$-dimensional maximally entangled
pairs that are shared between Alice and Bob in advance. The
communication should tolerate up to $d-1$ erasures introduced by the
channel.
 
An $[n,k,d;c]_q$ EACC coding scheme is depicted in
Fig. \ref{fig:coding_scheme}. The following elements are needed to
specify a scheme:
\begin{itemize}
	\item $n,d \in \mathbb{Z}_{\geq 1}, c \in \mathbb{Z}_{\geq 0}$ and $q\in \mathbb{Z}_{\geq 2}$;
	\item a finite set $\mathcal{M}$ with $\log_q |\mathcal{M}|=k$;
	\item encoders ${\rm ENC}_{m} \colon \mathcal{D}(\mathcal{H}_{A})\to \mathcal{D}(\otimes_{i\in[N]}\mathcal{H}_{Q_i})$, $m\in \mathcal{M}$;
	\item decoders ${\rm DEC}_{\mathcal{E}}$ for $\mathcal{E} \in
          \binom{[n]}{d-1}$. Each ${\rm DEC}_{\mathcal{E}}$ is a POVM
          $\{\Lambda_{\mathcal{E}}^{\hat{m}}\}_{\hat{m}\in
            \mathcal{M}}$ on the quantum systems $(Q_i\colon i\in
                [n]\setminus \mathcal{E})$ and $B$.
\end{itemize}
The scheme works as follows. 

\noindent {\bf Entanglement distribution:} Before the communication,
Alice and Bob share $c$ maximally entangled pairs, $A_1B_1, \dots,
A_cB_c$.  Let $A = A_1\cdots A_c$ collectively denote the entangled
resources initially held by Alice, and let $B = B_1\cdots B_c$
collectively denote the entangled resources initially held by
Bob. Each $A_sB_s$ for $s \in [c]$ is in a $q$-dimensional maximally
entangled state, i.e.,
\begin{align}
	\ket{\psi}_{A_sB_s} = \frac{1}{\sqrt{q}}\sum_{i\in \mathbb{F}_q} \ket{i}_{A_s} \ket{i}_{B_s},
\end{align}
where $\{\ket{i}_{A_s}\}_{i\in \mathbb{F}_q}$ and
$\{\ket{j}_{B_s}\}_{j\in \mathbb{F}_q}$ are sets of orthonormal basis
vectors for the Hilbert spaces $\mathcal{H}_{A_s}$ and
$\mathcal{H}_{B_s}$, respectively.

\noindent {\bf Encoding:} Alice wants to send a classical message $m \in
\mathcal{M}$. The encoder ${\rm ENC}_m$ (which is a complete positive
trace-preserving (CPTP) map) takes the input quantum system $A$ and
outputs a quantum state on the joint quantum systems $Q_1,Q_2,\dots,
Q_n$, each of size equal to one qudit. At this stage, the joint state
of $Q_1\cdots Q_n B_1\cdots B_{c}$ is denoted as $\sigma_m$. The
states of the systems $Q_1,\dots, Q_n$ are sent through the erasure
channel $\mathcal{N}$ which outputs the quantum systems $(Q_i\colon
i\in [N] \setminus \mathcal{E})$ and erases $(Q_i\colon i\in
\mathcal{E})$.

\noindent {\bf Decoding:} Bob receives $(Q_i\colon i \in [N]
\setminus \mathcal{E})$ from $\mathcal{N}$. Bob knows $\mathcal{E}$,
based on which he conducts a measurement defined by a POVM
$\{\Lambda_{\mathcal{E}}^{\hat{m}}\}_{\hat{m}\in \mathcal{M}}$ on
$(Q_i\colon i \in [n]\setminus \mathcal{E})$ and $B$.  The decoder
output $ \widehat{M}=\hat{m} $ has probability
$\Tr(\Lambda_{\mathcal{E}}^{\hat{m}}\sigma^m_{\mathcal{E}})$, where
$\sigma^m_{\mathcal{E}}$ denotes the partial state for $(Q_i\colon i
\in [n] \setminus \mathcal{E})$ and $B$. We require that a feasible scheme
must have
$\Tr(\Lambda_{\mathcal{E}}^{\hat{m}}\sigma^m_{\mathcal{E}})=\mathbb{I}(\hat{m}=m)$
for every $m,\hat{m} \in \mathcal{M}$, $\mathcal{E} \in
\binom{[n]}{d-1}$, i.e., we require perfect correction of the erasures.

\begin{figure}[htbp]
\center
\begin{tikzpicture}[scale=0.95, transform shape]
\node (T) at (-1,-2.5) {\small $\ket{\psi}_{AB}$};
 
\node (Q1) at (2.2,0) [ ] {$Q_1$};
\node (Q2) at (2.2,-0.8) [ ] {$Q_2$};
\node at (2.2,-1.4) [align=center] {$\vdots$};
\node (Q3) at (2.2,-2.2) [ ] {\small $Q_{n}$};

\node (Enc) at (1.1,-1.1) [draw, thick, rectangle,  minimum width = 1cm, minimum height = 3cm] {\small ${\rm ENC}_{m}$};

\node (Ch) at (3.1,-1.1) [draw, thick, rectangle, minimum height = 3cm, minimum width = 0.8 cm] {\small $\mathcal{N}$};

\node (Dec) at (5.8,-2.6) [draw, thick, rectangle, minimum height = 3 cm, minimum width = 1cm] {\small ${\rm DEC}_{\mathcal{E}}$};

\node (output) at (6.85,-2.75) {$\widehat{M}$};

\draw [color=gray, thick] ($(T.east)+(-0.1,0)$)--($(T.east)+(0.5,1.25)$)  node[pos=0.6, font=\small, left, text=black]{$A$}  --($(T.east)+(0.95,1.25)$);

\draw [color=gray, thick] ($(T.east)+(-0.1,0)$)--($(T.east)+(0.5,-1.25)$) node[pos=0.6, font=\small, left, text=black]{$B$}  --($(T.east)+(5.7,-1.25)$) node[pos=0.5, font=\small,above, text=black]{$B=B_1 \cdots B_{c}$};

\foreach \i in {1,2,3} {
\draw [color=gray, thick] ($(Q\i.east)+(-0.85,0)$)--($(Q\i.east)+(-0.65,0)$);
\draw [color=gray, thick] ($(Q\i.east)+(-0.1,0)$)--($(Q\i.east)+(0.15,0)$);
}

\draw [color=gray, thick] ($(Ch.east)+(0,0.5)$)--($(Ch.east)+(0.4,0.5)$)  node[font=\small, below right = -0.3cm and -0.1cm, align=center, text=black] {$ Q_i\colon  i\in \mathcal{E}$};

\draw [color=gray, thick] ($(Ch.east)+(0,-0.5)$)--($(Ch.east)+(0.4,-0.5)$) node[font=\small, below right = -0.3cm and  -0.4cm, align=center, text=black] {$Q_i\colon$ \\ $ i\in [n]\setminus \mathcal{E} $} ;

\draw [color=gray, thick] ($(Ch.east)+(1.2,-0.5)$)--($(Ch.east)+(1.75,-0.5)$);

\draw [color=gray, thick] (output) -- ($(output.west)+(-0.18,0)$);

\draw [color=gray, thick] ($(Q1.south)+(0,-3.8)$) -- ($(Q1.south)+(0,-4)$)--($(Q1.south)+(2.5,-4)$) node [below, pos=0.5, font=\small, text=black] {$\sigma_m$} --($(Q1.south)+(2.5,-3.8)$);

\end{tikzpicture}
\caption{General description of a coding scheme for EACC.  
}
\label{fig:coding_scheme}
\end{figure}
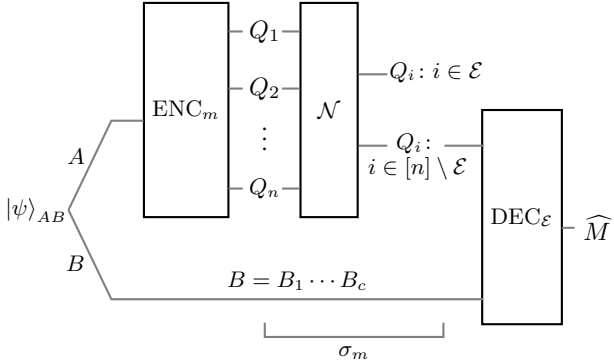

Throughout this letter, let us assume that $n,d\in \mathbb{Z}_{\geq
  1}$, $c \in \mathbb{Z}_{\geq 0}$ and $c \leq n$. We refer to such
parameters as admissible.  For any admissible $n$, $d$, $c$, define
\begin{alignat}{5} \label{eq:def_k_star}
  k^* = \sup \{k\colon \text{there} & \text{ exists $q$ and}\notag\\
  &\text{an $[n,k,d;c]_q$ coding scheme} \},
\end{alignat}
as the supremum of $k$ over all $q$ and $[n,k,d;c]_q$ coding schemes.

Recall the open problem in \cite{Grassl2025codes}, to determine
whether for all admissible $n$, $d$, $c$, the bound \eqref{eq:singleton} is
attained, i.e., to determine whether $k^* = (1+c/n) (n-d+1)$ for all
admissible $n$, $d$, $c$.

\begin{remark} Note that the Singleton bound \eqref{eq:singleton} does
  not depend on $q$, i.e., it holds for every $q$. Indeed Singleton
  bounds are entropic bounds, and scale-invariance is a typical
  feature of such bounds (recall that the entropic region is a
  cone). Since each $q$ corresponds to a different EACC channel, the
  bound applies to the entire class of channels. As it is evident from
  the definition of $k^*$ in \eqref{eq:def_k_star}, the question of
  tightness of the bound should be understood as the question:
  \emph{does there exist a channel in this class (i.e., does there
  exist a $q$) for which the bound is tight?} In particular, note that
  the tightness of the bound does not imply that it is tight for every
  $q$.
\end{remark}

\section{Saturation of the EACC Singleton Bound}
The following theorem follows from the observation in
\cite{sun2025capacityerasurepronequantumstorage}, that space sharing
suffices to answer the open question in \cite{Grassl2025codes}.
\begin{theorem} \label{thm:singleton}
	For any admissible $n$, $d$, $c$ with $c\leq n$, $k^* = (1+c/n) (n-d+1)$.
\end{theorem}
Let us note that Theorem \ref{thm:singleton} states that for any
admissible $n$, $d$, $c$, there exists a $q$ and a corresponding
$[n,k,d;c]_q$ coding scheme with $k$ saturating the EACC Singleton
bound.  Before we present the general proof, let us provide a toy
example.

\subsection{Example}
In this example, we show the construction of a $[3,10/3,1;2]_8$ EACC
code. Note that for $n=3$, $d=1$, $c=2$, this code saturates the EACC
Singleton bound with $q=8$.  First, $q=8$ means that we send
$8$-dimensional quantum systems over an erasure channel. Before the
communication, Alice and Bob share $c=2$ pairs of quantum systems,
namely, $A_1B_1$ and $A_2B_2$, where each pair is in an
$8$-dimensional maximally entangled state. In fact, for $i\in [2]$,
$A_iB_i$ can be equivalently viewed as $3$ pairs of maximally
entangled qubits, since $q=2^3$. Specifically, for $i\in [2]$, $A_i$
can be decomposed into subsystems $A_{i,1}$, $A_{i,2}$, $A_{i,3}$, and
similarly $B_i$ consists of subsystems $B_{i,1}$, $B_{i,2}$, $B_{i,3}$.
Each pair $A_{i,j}B_{i,j}$ is, without loss of generality, in the
maximally entangled state $\frac{1}{\sqrt{2}}
\big(\ket{0}_{A_{i,j}}\ket{0}_{B_{i,j}}+
\ket{1}_{A_{i,j}}\ket{1}_{B_{i,j}}\big)$ for $i\in [2]$, $j\in [3]$.

Next, we conduct an entanglement rearrangement operation to spread the
entanglement at the transmitter uniformly into
$Q_1,Q_2,Q_3$. Specifically, consider that for $i\in [3]$, the quantum
system $Q_i$ to be sent through the channel consists of subsystems
$Q_{i,1}, Q_{i,2}, Q_{i,3}$, where each $Q_{i,j}$ is a qubit. We first
initialize $Q_{1,1} Q_{2,1} Q_{3,1}$ to the state
$\ket{0}_{Q_{1,1}}\ket{0}_{Q_{2,1}}\ket{0}_{Q_{3,1}}$. Then make the
following assignment:
\begin{align} \label{eq:ER}
	\begin{bmatrix}
		Q_{1,2} & Q_{2,2} & Q_{3,2} \\
		Q_{1,3} & Q_{2,3} & Q_{3,3}
	\end{bmatrix}
	=
	\begin{bmatrix}
		A_{1,1} & A_{1,3} & A_{2,2} \\
		A_{1,2} & A_{2,1} & A_{2,3}
	\end{bmatrix}.
\end{align}
Note that for $i\in [3]$, each $Q_i = Q_{i,1}Q_{i,2}Q_{i,3}$ has $2/3$
parts that are maximally entangled with Bob's quantum resource.

Depending on the message $m$, Alice processes $Q_1,Q_2,Q_3$. For this
part, we apply the idea of space-sharing of
\cite{sun2025capacityerasurepronequantumstorage}. To apply the
space-sharing argument in this example, we will first construct three
EACC codes, $C_1$, $C_2$, and $C_3$, each with $n=3$, $d=2$, and
$q=2$. Moreover, $C_j$ is used for encoding the subsystems
$Q_{1,j}Q_{2,j}Q_{3,j}$ for $j\in [3]$. We then combine the three
codes into a (bigger) EACC code with the same $n=3$ and $d=2$, but
$q=8$.

Specifically, $C_1$ is a $[3,2,2;0]_2$ EACC code which does not use
entanglement. We use a classical binary MDS code with generator matrix
\begin{align}
	G = \begin{bmatrix}
		1 & 0 & 1\\
		0 & 1 & 1
	\end{bmatrix} \in \mathbb{F}_2^{2\times 3}.
\end{align}
The encoding process of $C_1$ is to map two information bits
$[u_1,u_2]$ into a coded vector $[x_1,x_2,x_3] = [u_1,u_2]G$, and
directly encode each $x_i$ into $\ket{x_i}_{Q_{i,1}}$ for $i \in
\{1,2,3\}$, where $\{\ket{i}\}_{Q_{i,1}}$ denotes a set of orthonormal
basis vectors for the Hilbert space
$\mathcal{H}_{Q_{i,1}}$. Mathematically, the quantum state of
$Q_{1,1}Q_{2,1}Q_{3,1}$ after this process is
\begin{align}
	\ket{x_1}_{Q_{1,1}} \ket{x_2}_{Q_{2,1}} \ket{x_3}_{Q_{3,1}}.
\end{align}

The code $C_2$ is a $[3,4,2;3]_2$ EACC code.  We construct it based
on the same matrix $G$, now considered as the generator matrix of a
classical code over $\mathbb{F}_4$, and the superdense coding protocol
\cite{bennett1992communication}. Let $[u_3,u_4,u_5,u_6]$ be four
information bits.  Let $[y_1,y_2,y_3] = [u_3,u_4]G$, $[y_1',y_2',y_3']
= [u_5,u_6]G$.  According to \eqref{eq:ER}, we have $Q_{1,2}B_{1,1},
Q_{2,2}B_{1,3}, Q_{3,2}B_{2,2}$ each being a pair of maximally
entangled qubits. For $i\in [3]$, depending on the pair of bits
$[y_i,y_i']$ Alice applies a local unitary transformation
$X^{y_i}Z^{y_i'}$ which maps the maximally entangled state into one of
the four mutually orthogonal states of the Bell basis.  The the
resulting quantum state is
\begin{align}
	\ket{\beta_{y_1,y_1'}}_{Q_{1,2}B_{1,1}} \ket{\beta_{y_2,y_2'}}_{Q_{2,2}B_{1,3}}\ket{\beta_{y_3,y_3'}}_{Q_{3,2}B_{2,2}},
\end{align}
where $\{\beta_{{\sf x},{\sf z}}\}_{{\sf x}\in \{0,1\}, {\sf z}\in \{0,1\}}$ denote the set of Bell basis vectors. 

The code $C_3$ is the same as the code $C_2$, with its $4$ information
bits denoted as $[u_7,u_8,u_9,u_{10}]$. Let $[z_1,z_2,z_3] =
[u_7,u_8]G$, $[z_1',z_2',z_3'] = [u_9,u_{10}]G$. Similar to the
encoding process of $C_2$, Alice encodes $[z_i,z_i']$ for $i\in [3]$
into $Q_{1,3}B_{1,2},Q_{2,3}B_{2,1},Q_{3,3}B_{2,3}$, respectively. The
resulting quantum state is
\begin{align}
	\ket{\beta_{z_1,z_1'}}_{Q_{1,3}B_{1,2}} \ket{\beta_{z_2,z_2'}}_{Q_{2,3}B_{2,1}}\ket{\beta_{z_3,z_3'}}_{Q_{3,3}B_{2,3}}.
\end{align}

For the decoding process, suppose one of $Q_1,Q_2,Q_3$ is erased. Due
to symmetry in the following we assume without loss of generality that
$Q_3$ is erased. For the code $C_1$, Bob measures $Q_{1,1}, Q_{2,1}$
in their respective computational basis and obtains $x_1,x_2$, from
which he decodes $u_1,u_2$ since $G$ is MDS. For the code $C_2$, Bob
measures $Q_{1,2}B_{1,1}$ and $Q_{2,2}B_{1,3}$ in the Bell basis and
obtains $[y_1,y_1',y_2,y_2']$, from which he decodes
$[u_3,u_4,u_5,u_6]$. Correspondingly, for the code $C_3$, Bob measures
$Q_{1,3}B_{1,2}$ and $Q_{2,3}B_{2,1}$ in the Bell basis and obtains
$[z_1,z_1',z_2,z_2']$, from which he decodes
$[u_7,u_8,u_9,u_{10}]$. Therefore, the combined code $C$, which is the
combination of $C_1$, $C_2$ and $C_3$, guarantees the decoding of
$[u_1,u_2,\cdots, u_{10}]$, in total $10$ bits of classical
information, while tolerating $1$ erasure of the quantum system $Q_1$,
$Q_2$ or $Q_3$.  The total number of classical information transmitted
is $10$ bits, and therefore $k=\log_q 2^{10} = \log_8 1024 =
10/3$. Finally, $C$ has $n=3$ and $d=2$, the same as those for the
codes $C_1$ ,$C_2$, and $C_3$. Therefore, the combined code $C$ is a
$[3,10/3,2;2]_8$ EACC code.  We are using the three codes for qubit
channels ($q=2$) simultaneously over a channel with $q=2^3=8$.

\subsection{Proof of Theorem \ref{thm:singleton}}
The extension to the general case is straightforward. Recall that we
are given admissible $n$, $d$, $c$. Let $r = n/\gcd(n,c)$, $\ell_1 =
\tfrac{(n-c)r}{n}$, and $\ell_2 = \tfrac{cr}{n}$. Clearly, $\ell_i \in
\mathbb{Z}_{\geq 0}$ for $i\in \{1,2\}$ and that
$\ell_1+\ell_2=r$. Let $k_1=n-d+1$ and $k_2=2(n-d+1)$.  Suppose
$\bar{q}$ is a power of a prime and $\bar{q}\geq n$. Then there exists
a classical $[n,k_1,n-k_1+1]_{\bar{q}}$ MDS code (e.g., a Reed-Solomon
code) over $\mathbb{F}_{\bar{q}}$. Let $q = \bar{q}^r \in
\mathbb{Z}_{\geq 0}$.  Let $C_1,C_2,\dots, C_{\ell_1}$ each be an
$[n,k_1,n-k_1+1;0]_{\bar{q}}$ EACC code, i.e., constructed from a
classical $[n,k_1,n-k_1+1]_{\bar{q}}$ MDS code over
$\mathbb{F}_{\bar{q}}$.

Similar as in the example above, one can construct an
$[n,k_2,d;n]_{\bar{q}}$ EACC code by utilizing $n$ maximally entangled
pairs of dimension $\bar{q}$ and the superdense coding protocol (see
also \cite{brun2014catalytic}).
Since Alice and Bob share $c$ maximally entangled pairs of dimension
$q$, which are equivalently $cr$ maximally entangled pairs of
dimension $\bar{q}$, they are able to construct $\tfrac{cr}{n} =
\ell_2$ such $[n,k_2,d;n]_{\bar{q}}$ EACC codes. Let us denote them by
$C_{\ell_1+1},C_{\ell_2+1},\dots, C_{\ell_1+\ell_2}$.

Now combining all codes $C_1,C_2,\dots, C_{\ell_1+\ell_2}$ using the
space-sharing argument, we obtain an $[n,k,d;c]_q$ EACC code $C$,
where $k = \log_q (\bar{q}^{k_1\ell_1 + k_2\ell_2})
=\frac{k_1\ell_1+k_2\ell_2}{\ell_1+\ell_2} = (1+c/n)(n-d+1)$,
attaining the EACC Singleton bound. \hfill \qed

\subsection{Asymptotic saturation with large $q$}
We conclude this section by pointing out that given admissible $n$,
$d$, $c$, for sufficiently large $q$, there are EACC codes with $k\to
(1+c/n)(n-d+1)$ as $q \to \infty$. This is because for sufficiently
large $q$, there exists $\bar{q} \in [\tfrac{q^{1/r}}{2},q^{1/r}]$
such that $\bar{q}\geq n$ and $\bar{q}$ is a power of $2$. One can
then apply the same reasoning to construct such $\ell_1+\ell_2$ EACC
codes for dimension $\bar{q}$ and obtain a combined EACC code for
dimension $\widetilde{q}$, where $\widetilde{q} = \bar{q}^{r}$ and thus
$\frac{q}{2^r}\leq \widetilde{q} \leq q$. Since we are using
a quantum erasure channel for $q$-dimensional systems, it allows
sending $\widetilde{q}$-dimensional systems (for $\widetilde{q}\le q$,
not using the remaining dimensions). The corresponding $k$ for the
combined code is $k = \log_q (\bar{q}^{k_1\ell_1 + k_2\ell_2}) =
\log_q (\widetilde{q}) \cdot \log_{\widetilde{q}}(\bar{q}^{k_1\ell_1 +
  k_2\ell_2}) \geq \log_q(\tfrac{q}{2^r})
\cdot\frac{k_1\ell_1+k_2\ell_2}{\ell_1+\ell_2} = (1-r\log_q 2)
(1+c/n)(n-d+1)$, which approaches $(1+c/n)(n-d+1)$ as $q\to \infty$.

\section{EACC with Separate Encoders}
The framework in Fig. \ref{fig:coding_scheme} allows \emph{joint}
(quantum) processing of the systems $A_1,A_2,\dots, A_c$ carrying half
of the maximally entangled states on $A_iB_i$ to produce
$Q_1,Q_2,\dots, Q_n$.  The same applies to the receiver, jointly
processing $Q_{[n]\setminus\mathcal{E}}$ and $B_1,B_2,\dots, B_c$ to
produce $\widehat{M}$.  Such joint processing may not always be
feasible.  Such a scenario is, for example, distributed storage where
some of the nodes have pre-shared entanglement, but there is no
quantum communication between the nodes when (classical) information
is stored.

Therefore, it is also important to study codes without such
joint processing, and in particular the performance limit of such
codes. To this end, let us impose a distribution processing constraint
on the EACC framework and consider the class of coding schemes
illustrated in Fig. \ref{fig:separate_enc}. Specifically, for $i\in
\{1,2,\cdots, c\}$, $\mbox{ENC}_m^i$ is a separate encoder, which
outputs $Q_i$ only based on $m$ and $A_i$, and for $i\in
\{c+1,c+2,\cdots,n\}$, $\mbox{ENC}_m^i$ is again a separate encoder
that outputs $Q_i$ only based on $m$.

Notably, the coding schemes in Fig. \ref{fig:separate_enc} subsume the
coding schemes considered in \cite{Grassl2025codes}. Our next theorem
establishes an entropic Singleton bound for any $[n,k,d;c]_q$ EACC
coding scheme with separate encoders.
 
\begin{figure}[h]
\center
\begin{tikzpicture}[scale=0.94, transform shape]
\node (T) at (-1,-2.5) {\small $\ket{\psi}_{AB}$};
 
\node (Q1) at (2.85,0.5) [ ] {$Q_1$};
\node (Q2) at (2.85,-0.8) [ ] {$Q_c$};
\node (Q3) at (2.85,-1.5) [ ] {\small $Q_{c+1}$};
\node (Q4) at (2.85,-2.75) [ ] {\small $Q_{n}$};

\node at (0.6,0.1) [align=center] {$\vdots$};

\node at (1.7,0) [align=center] {$\vdots$};
\node at (1.7,-2) [align=center] {$\vdots$};

\node at (2.85,-0.) [align=center] {$\vdots$};
\node at (2.85,-2) [align=center] {$\vdots$};

\node (Enc1) at (1.65,0.5) [draw, thick, rectangle,  minimum width = 1.4cm, minimum height = 0.6cm] {\footnotesize ${\rm ENC}_{m}^1$};

\node (Enc2) at (1.65,-0.75) [draw, thick, rectangle,  minimum width = 1.4cm, minimum height = 0.6cm] {\footnotesize ${\rm ENC}_{m}^c$};

\node (Enc3) at (1.65,-1.5) [draw, thick, rectangle,  minimum width = 1.4cm, minimum height = 0.6cm] {\footnotesize ${\rm ENC}_{m}^{c+1}$};

\node (Enc4) at (1.65,-2.75) [draw, thick, rectangle,  minimum width = 1.4cm, minimum height = 0.6cm] {\footnotesize ${\rm ENC}_{m}^n$};

\node (Ch) at (3.75,-1.1) [draw, thick, rectangle, minimum height = 3.8cm, minimum width = 0.8 cm] {\small $\mathcal{N}$};

\node (Dec) at (6.5,-2.6) [draw, thick, rectangle, minimum height = 3 cm, minimum width = 1cm] {\small ${\rm DEC}_{\mathcal{E}}$};

\node (output) at (7.55,-2.75) {$\widehat{M}$};

\draw [color=gray, thick] ($(T.east)+(-0.1,0)$)--($(T.east)+(0.5,3)$)   --($(T.east)+(1.4,3)$) node[pos=0.6, font=\small, above=-0.1, text=black]{$A_1$};

\draw [color=gray, thick] ($(T.east)+(-0.1,0)$)--($(T.east)+(0.5,1.75)$)   --($(T.east)+(1.4,1.75)$) node[pos=0.6, font=\small, above=-0.1, text=black]{$A_c$} ;

\draw [color=gray, thick] ($(T.east)+(-0.1,0)$)--($(T.east)+(0.5,-1.25)$)    --($(T.east)+(6.4,-1.25)$) node[pos=0.3, font=\small,above, text=black]{$B=B_1 \cdots B_{c}$};

\foreach \i in {1,2, 4} {
\draw [color=gray, thick] ($(Q\i.east)+(-0.85,0)$)--($(Q\i.east)+(-0.65,0)$);
\draw [color=gray, thick] ($(Q\i.east)+(-0.1,0)$)--($(Q\i.east)+(0.15,0)$);
}
\draw [color=gray, thick] ($(Q3.east)+(-1,0)$)--($(Q3.east)+(-0.85,0)$);
\draw [color=gray, thick] ($(Q3.east)+(-0.1,0)$)--($(Q3.east)+(0.03,0)$);

\draw [color=gray, thick] ($(Ch.east)+(0,0.5)$)--($(Ch.east)+(0.4,0.5)$)  node[font=\small, below right = -0.3cm and -0.1cm, align=center, text=black] {$ Q_i\colon  i\in \mathcal{E}$};

\draw [color=gray, thick] ($(Ch.east)+(0,-0.5)$)--($(Ch.east)+(0.4,-0.5)$) node[font=\small, below right = -0.3cm and  -0.4cm, align=center, text=black] {$Q_i\colon$ \\ $ i\in [n]\setminus \mathcal{E} $} ;

\draw [color=gray, thick] ($(Ch.east)+(1.2,-0.5)$)--($(Ch.east)+(1.75,-0.5)$);

\draw [color=gray, thick] (output) -- ($(output.west)+(-0.18,0)$);

\draw [color=gray,thick] ($(Q1.south)+(0,-4.2)$) -- ($(Q1.south)+(0,-4.4)$)--($(Q1.south)+(2.5,-4.4)$) node [below, pos=0.5, font=\small, text=black] {$\sigma_m$} --($(Q1.south)+(2.5,-4.2)$);

\end{tikzpicture}
\caption{EACC coding framework with separate encoders. }
\label{fig:separate_enc}
\end{figure}
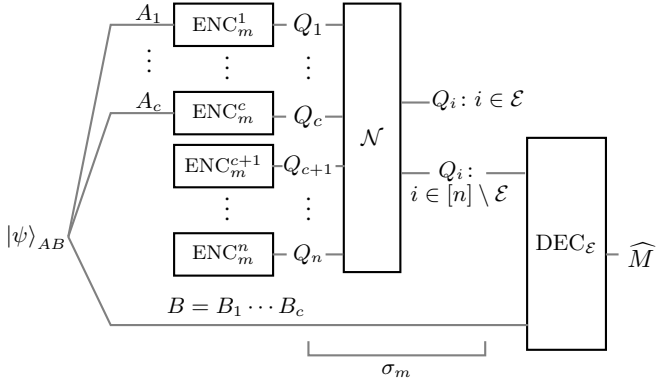

\begin{theorem} \label{thm:separate}
	For any EACC code with separate encoders,
	\begin{align} \label{eq:separate}
		k &{}\leq \max\{n+c-2d+2, n-d+1\} \notag \\
		&{}= \begin{cases}
		    n+c-2d+2, & 0 \leq d-1 \leq c;\\
			n-d+1, & c\leq d-1 \leq n.
		\end{cases}
	\end{align}
\end{theorem}
One can verify that the bound in \eqref{eq:separate} is in general
less than or equal to the bound in \eqref{eq:singleton}.  We point out
that, given any admissible $n$, $d$, $c$, there exist $q$ and corresponding
coding schemes with separate encoders that saturate the bound
\eqref{eq:separate}, i.e., the bound is tight.  Reference
\cite{Grassl2025codes} further discusses aspects related to the
channel dimension $q$.
\begin{proof}
Given any coding scheme with separate encoders, since it must
guarantee correct decoding, i.e., $\Pr(\widehat{M}=m)=1$, for every
realization of the message $m\in \mathcal{M}$, it must guarantee
correct decoding if the message is a random variable $M$ uniformly
distributed over $\mathcal{M}$.  Define the joint classical-quantum
state
\begin{align}
	\sigma_{MQ_1\cdots Q_n} = \frac{1}{|\mathcal{M}|}\sum_{m\in \mathcal{M}}\ket{m}\bra{m}_M\otimes \sigma_m,
\end{align}
where $\{\ket{m}\}_{m\in \mathcal{M}}$ denotes a set of orthonormal
basis vectors in the $|\mathcal{M}|$-dimensional Hilbert space
associated with the classical random variable $M$.  For notational
simplicity, we will write the state as $\sigma$.  In the following,
for a subset $I\subseteq [n]$, we will use $Q_{I}$ to denote the
quantum system composed of $(Q_{i})_{i\in I}$.  All information
quantities will be defined with respect to logarithms base $q$.  Let
us consider two regimes.
	
First, for $0\leq d-1\leq c$, let $I= \{d,d+1,\ldots, c\}$, $J=
\{c+1,c+2,\ldots, n\}$. We have $\big|[n] \setminus (I \cup J)\big| =
d-1$. Consider the decoder ${\rm DEC}_{\mathcal{E}}$ with
$\mathcal{E}=[n] \setminus (I \cup J)$. Note that the input to this
decoder is $Q_IQ_JB$. We have that {\small
	\begin{align}
	k &= H(M)\notag\\ 
	&= I(M;\widehat{M}) \label{eq:bound_1_decoding} \\
		&\leq I(M;Q_I  Q_J B)_{\sigma} \label{eq:bound_1_holevo} \\
		&= I(M;Q_I Q_J | B)_{\sigma} \label{eq:bound_1_nosig} \\
		&= H(Q_I Q_J| B)_{\sigma} - H(Q_IQ_J | M  B)_{\sigma} \label{eq:bound_1_def_cMI} \\
		&\leq (n-d+1) - H(Q_I| M B)_{\sigma} - H(Q_J | Q_I  M B)_{\sigma} \label{eq:bound_1_dim1_and_chain} \\
		&= (n-d+1)  - H(Q_I | MB)_{\sigma} - H(Q_J | M)_{\sigma}  \label{eq:bound_1_indep}  \\
		&\leq (n-d+1)  - H(Q_I | MB)_{\sigma} \label{eq:bound_1_classical} \\
		&\leq (n-d+1)  + (c-d+1) \label{eq:bound_1_dim2} \\
		&=n+c-2d+2.
	\end{align}
}%
Step \eqref{eq:bound_1_decoding} is because $\Pr(M=\widehat{M})=1$.
Step \eqref{eq:bound_1_holevo} is from  Holevo's bound \cite{holevo1973bounds}. 
Step \eqref{eq:bound_1_nosig} is because $B$ is independent of $M$ according to the no-communication theorem \cite{peres2004quantum}.  
Step \eqref{eq:bound_1_def_cMI} is from the definition of conditional mutual information. 
Step \eqref{eq:bound_1_dim1_and_chain} is because $H(Q_IQ_J| B)_{\sigma} \leq H(Q_IQ_J)_{\sigma}$ as conditioning does not increase entropy, and then $H(Q_IQ_J)\leq \log_q |Q_IQ_J| = n-d+1$. 
Step \eqref{eq:bound_1_indep} is because conditioned on any $M=m$, $Q_J$ is independent of $Q_I B$ as $Q_J$ only depends on the message. 
Step \eqref{eq:bound_1_classical} follows from the non-negativity of conditional entropy when the conditioning system is classical. 
Step \eqref{eq:bound_1_dim2} is due to the weak monotonicity of quantum entropy, so $-H(Q_I| MB)_{\sigma} \leq H(Q_I)_{\sigma}$, which is then upper bounded by $\log_q|Q_I| =c-d+1$.

Second, for $c\leq d-1 \leq n$, let $J= \{d,d+1,\ldots, n\}$. We have
$\big|[n]\setminus J \big| = d-1$. Consider the decoder ${\rm
  Dec}_{\mathcal{E}}$ with $\mathcal{E} = [n] \setminus J$. Note that
the input to this decoder is $Q_J B$. We have that {\small
\begin{align}
	k &= H(M)\notag \\
	&= I(M;\widehat{M}) \label{eq:bound_2_decoding} \\
	&\leq I(M;Q_J B)_{\sigma} \label{eq:bound_2_holevo} \\
	&=I(M;Q_J| B)_{\sigma} \label{eq:bound_2_nosig} \\
	&= H(Q_J| B)_{\sigma} - H(Q_J| MB)_{\sigma} \label{eq:bound_2_def_cMI} \\
	&\leq n-d+1 - H(Q_J| MB)_{\sigma} \label{eq:bound_2_dim} \\
	&= n-d+1 - H(Q_J| M)_{\sigma} \label{eq:bound_2_indep} \\
	&\leq n-d+1. \label{eq:bound_2_classical}
\end{align}
}%
Step \eqref{eq:bound_2_decoding} is because $\Pr(M=\widehat{M}) = 1$.
Step \eqref{eq:bound_2_holevo} is from  Holevo's bound. Step \eqref{eq:bound_2_nosig} is because $B$ is independent of $M$.
Step \eqref{eq:bound_2_def_cMI} is from the definition of conditional mutual information.
Step \eqref{eq:bound_2_dim} is because $H(Q_J| B)_{\sigma} \leq H(Q_J)_{\sigma}$ as conditioning does not increase entropy, and then $H(Q_J) \leq \log_q |Q_J| = n-d+1$.
Step \eqref{eq:bound_2_indep} is because conditioned on any $M=m$, $Q_J$ is independent of $B$.
Step \eqref{eq:bound_2_classical} follows from the non-negativity of conditional entropy when the conditioning system is classical.
\end{proof}

\section{Discussion}
In this letter, we revisited the EACC coding framework of
\cite{Grassl2025codes}. Firstly, following from an observation in
\cite{sun2025capacityerasurepronequantumstorage} we presented a
solution to the open problem in \cite{Grassl2025codes} on the
tightness of the EACC Singleton bound for general EACC coding
schemes. We showed that for any admissible $n$, $d$, $c$, there exist
(in fact infinitely many) $q$ and EACC coding schemes for which the
EACC Singleton bound is achieved. Secondly, we studied a constrained
class of EACC coding schemes, explicitly those with entanglement
assistance distributed among separate encoders, and we derived a new
entropic Singleton bound for this setting. We noted that the tightness
of this bound follows from the coding schemes presented in
\cite{Grassl2025codes}.  We conclude by noting that EACC codes
obtained via the space-sharing argument may require a very large
channel dimension $q$ in certain cases. The study of EACC codes
attaining the Singleton bound with potentially smaller $q$ remains an
interesting open direction.

\bibliographystyle{IEEEtran}
\bibliography{../bib_file/yy.bib}

\end{document}